\documentclass[11pt]{article}
\usepackage[a4paper,margin=1in]{geometry}
\usepackage[T1]{fontenc}
\usepackage[utf8]{inputenc}
\usepackage{lmodern}
\usepackage{microtype}
\usepackage{amsmath,amsthm,amssymb,mathtools}
\usepackage{enumitem}
\usepackage{hyperref}
\usepackage[nameinlink,noabbrev]{cleveref}
\usepackage{mathrsfs}
\usepackage[backend=biber,style=numeric,sorting=none]{biblatex}
\hypersetup{colorlinks=true,linkcolor=blue,citecolor=blue,urlcolor=blue}

\usepackage{aliascnt}

\newtheorem{theorem}{Theorem}[section]

\newaliascnt{proposition}{theorem}
\newtheorem{proposition}[proposition]{Proposition}
\aliascntresetthe{proposition}
\crefname{proposition}{proposition}{propositions}
\Crefname{proposition}{Proposition}{Propositions}

\newaliascnt{lemma}{theorem}
\newtheorem{lemma}[lemma]{Lemma}
\aliascntresetthe{lemma}
\crefname{lemma}{lemma}{lemmas}
\Crefname{lemma}{Lemma}{Lemmas}

\newaliascnt{corollary}{theorem}
\newtheorem{corollary}[corollary]{Corollary}
\aliascntresetthe{corollary}
\crefname{corollary}{corollary}{corollaries}
\Crefname{corollary}{Corollary}{Corollaries}

\newaliascnt{problem}{theorem}

\aliascntresetthe{problem}
\crefname{problem}{problem}{problems}
\Crefname{problem}{Problem}{Problems}

\newaliascnt{definition}{theorem}
\newtheorem{definition}[definition]{Definition}
\aliascntresetthe{definition}
\crefname{definition}{definition}{definitions}
\Crefname{definition}{Definition}{Definitions}

\newaliascnt{remark}{theorem}

\aliascntresetthe{remark}
\crefname{remark}{remark}{remarks}
\Crefname{remark}{Remark}{Remarks}

\newcommand{\tr}{\operatorname{Tr}}

\newcommand{\sgn}{\operatorname{sign}}
\newcommand{\C}{\mathbb{C}}
\newcommand{\R}{\mathbb{R}}

\newcommand{\ket}[1]{|#1\rangle}
\newcommand{\bra}[1]{\langle #1|}
\newcommand{\lp}{\mathtt{p}}
\newcommand{\lt}{\mathtt{t}}

\title{The State Cost of Classical Simulation of One-Way General Quantum Finite Automata}
\author{Zeyu Chen and Junde Wu}
\date{}

\begin{document}
\maketitle

\begin{abstract}
Under strict cutpoints, probabilistic finite automata (PFAs) and one-way general quantum finite automata (1gQFAs) recognize the same stochastic languages, shifting the theoretical focus to the state cost required for a classical PFA to simulate a 1gQFA. For an $n$-state ($n\geq 2$) 1gQFA, the state cost upper bound of classical simulation was previously known to be $n^2+3$, while the lower bound remained an open problem widely conjectured to be quadratic. After establishing a well-defined notion of classical simulation, we improve the existing state cost upper bound from $n^2+3$ to $n^2+1$. Subsequently, we introduce the concepts of shattering and dynamic shattering, which are used to determine the memory required for a PFA to recognize a language. Using these techniques, we prove that the state cost lower bound of the simulation reaches $n^2+1$ over a four-letter alphabet. With the upper and lower bounds thus matching, the problem is fully resolved.
\end{abstract}


\section{Introduction}
Finite automata provide a clean setting for separating computational power from memory cost: the internal memory is finite, the transition rules are explicit, and the number of states is the standard descriptional resource. Quantum finite automata (QFAs) replace classical transitions by quantum evolution and measurement, so quantum effects can change both what is recognizable and how compactly a computation can be represented \cite{moore2000,kondacs1997,sayyakaryilmaz2014,ambainisyakaryilmaz2021}. Finite-state ideas also appear in concrete quantum implementations and control-oriented constructions, including photonic QFA realizations and quantum Braitenberg vehicles built from finite automata \cite{mereghetti2020,mishra2020}.

For one-way quantum automata, comparisons with classical finite automata naturally separate three questions. \emph{Recognition power} asks which languages a model can recognize. \emph{State succinctness} asks, on languages available to both models, how compactly the same task can be represented. In the bounded-error setting, broader finite-state models can exceed the recognition power of PFAs: two-way QFAs can recognize certain nonregular languages in polynomial time, while affine finite automata are strictly more powerful than PFAs in language recognition \cite{kondacs1997,diazcaroyakaryilmaz2016}. By contrast, restricted measure-once and measure-many one-way QFA models exhibit recognition restrictions while still achieving substantial state savings on suitable language families \cite{kondacs1997,ambainis1998,brodsky2002}, and one-way semi-quantum and hybrid quantum--classical variants exhibit further state tradeoffs and exponential succinctness \cite{zheng2014,qiu2015,xiaoqiu2025}. These results separate two sources of quantum--classical difference: the models may recognize different language classes, or quantum states may represent a shared language more compactly. When a model offers neither a recognition advantage nor exponential succinctness, a third question remains: how many classical states are required in the worst case to replace a quantum automaton while preserving its recognized language?

The strict-cutpoint regime isolates this \emph{classical simulation cost} particularly cleanly for 1gQFAs. Under strict-cutpoint recognition, 1gQFAs and PFAs recognize exactly the same class of stochastic languages \cite{yakaryilmaz2011,li2012,turakainen1969}. The simulation cost therefore isolates a purely descriptional question: how much probabilistic finite-state memory is needed to replace quantum mixed-state memory without changing the language? Determining the exact cost identifies which part of the quantum state-space geometry becomes an irreducible classical memory requirement and turns the asymptotic separation into an exact state bound.

Mixed-state linearization gives an $n^2$-state GFA representation of every $n$-state 1gQFA \cite{yakaryilmaz2011,li2012}. Combined with the standard GFA--PFA correspondence, this yields the previously known upper bound of $n^2+3$ PFA states \cite{turakainen1975,hirvensaloyakaryilmaz2016,rote2025}. In this paper, we sharpen the zero-cutpoint conversion, showing that a $k$-state zero-cutpoint GFA requires only $k+1$ PFA states. This improves the universal upper bound for 1gQFAs to $n^2+1$. 

The matching lower bound remained an open problem and was widely conjectured to be quadratic. To resolve it, we introduce shattering and dynamic shattering, which translate prefix--suffix acceptance patterns into lower bounds on PFA state complexity, and construct an $n$-state 1gQFA over a fixed four-letter alphabet that realizes the required shattering structure by exploiting the geometry of mixed quantum states and persistent quantum oscillation. This gives the matching lower bound $n^2+1$, and hence
\[
C_n=n^2+1\qquad(n\ge2).
\]

The remainder of the paper is organized as follows. \Cref{sec:prelim} introduces the preliminaries and the PFA, GFA, and 1gQFA models used throughout. \Cref{sec:upper} develops the improved GFA-to-PFA conversion and derives the $n^2+1$ universal upper bound. \Cref{sec:lower} introduces shattering and dynamic shattering and constructs the fixed-alphabet 1gQFA attaining the matching lower bound. \Cref{sec:conclusion} concludes the paper and outlines directions for future work.

\section{Automaton Models and State Complexity}\label{sec:prelim}
\subsection{Conventions and state complexity}
Throughout, $\Sigma$ is a finite input alphabet and $\Sigma^\ast$ is the free monoid over $\Sigma$. If $w=\sigma_1\cdots\sigma_\ell\in\Sigma^\ast$, then $|w|=\ell$. For a positive integer $m$, the probability simplex is denoted by
\[
\Delta^{m-1}=\{x\in\R^m:x_i\ge 0\text{ for all }i,\ \sum_{i=1}^m x_i=1\}.
\]
All probabilistic computations are written in row-vector form, so that stochastic matrices act on the right.

Throughout, languages are recognized with \emph{strict cutpoints}. For a machine $\mathcal A$ over $\Sigma$ with a real-valued output function $f_{\mathcal A}:\Sigma^\ast\to\R$, we define
\[
L(\mathcal A,\lambda)=\{w\in\Sigma^\ast:f_{\mathcal A}(w)>\lambda\}.
\]
For PFAs and QFAs, $f_{\mathcal A}$ is the acceptance-probability function. This is the usual unbounded-error notion \cite{rabin1963,paz1971,yakaryilmaz2010,yakaryilmaz2011}. Using this strict-cutpoint recognition, we define the classical simulation cost.

\begin{definition}
For each positive integer $n$, the \emph{worst-case classical simulation cost} $C_n$ is defined as the smallest integer such that for any $n$-state 1gQFA $Q$ and strict cutpoint $\lambda\in[0,1)$, there exists a one-way PFA $P$ with at most $C_n$ states and a strict cutpoint $\mu\in[0,1)$ satisfying
\[
L(P,\mu)=L(Q,\lambda).
\]
\end{definition}

Let $\mathcal H\cong\C^n$ be an $n$-dimensional complex Hilbert space. The set of all linear operators on $\mathcal H$, denoted by $\mathcal L(\mathcal H)$, forms a complex vector space equipped with the Hilbert--Schmidt inner product $\langle X,Y\rangle=\tr(X^\dagger Y)$. An operator $X\in\mathcal L(\mathcal H)$ is \emph{Hermitian} if $X=X^\dagger$. The collection of all Hermitian operators, denoted by $\operatorname{Herm}(\mathcal H)$, forms a real vector space. A \emph{density operator} $\rho$ is a positive semidefinite Hermitian operator of trace $1$. The real vector space $\operatorname{Herm}(\mathcal H)$ has dimension $n^2$. Its trace-one affine hyperplane has dimension $n^2-1$, and the density operators have this hyperplane as their affine hull \cite{watrous2018}.

\subsection{Automaton models}
A \emph{generalized finite automaton} (GFA) over $\Sigma$ is a tuple
\[
G=(S,\Sigma,u,\{A_\sigma\}_{\sigma\in\Sigma},v),
\]
where $S$ is a finite state set with $|S|=k$, $u\in\R^{1\times k}$ is an initial row vector, each $A_\sigma\in\R^{k\times k}$ is a real transition matrix, and $v\in\R^{k\times 1}$ is a final column vector, all indexed by $S$. For a word $w=\sigma_1\cdots\sigma_m$, its value is
\[
f_G(w)=uA_{\sigma_1}\cdots A_{\sigma_m}v.
\]
Given a strict cutpoint $\lambda\in\R$, the recognized language is
\[
L(G,\lambda)=\{w\in\Sigma^\ast:f_G(w)>\lambda\}.
\]
This model is standard in the theory of stochastic languages \cite{turakainen1969,paz1971}. Turakainen's classical results imply that GFAs and PFAs recognize the same family of strict-cutpoint languages.

A \emph{one-way probabilistic finite automaton} (PFA) is taken in the standard end-marker model \cite{rabin1963,paz1971}. A PFA over $\Sigma$ is a tuple
\[
P=(S,\Sigma,p_0,\{P_\sigma\}_{\sigma\in\Sigma},P_{\#},F),
\]
where $S$ is a finite state set, $p_0$ is an initial distribution on $S$, each $P_\sigma$ is a row-stochastic transition matrix on $S$, $P_{\#}$ is a row-stochastic end-marker matrix on $S$, and $F\subseteq S$ is the accepting set. If $w=\sigma_1\cdots\sigma_m$, then
\[
f_P(w)=p_0 P_wP_{\#}\mathbf 1_F,
\qquad
P_w=P_{\sigma_1}\cdots P_{\sigma_m},
\]
where $\mathbf 1_F$ is the indicator column vector of $F$. Under a strict cutpoint $\mu\in[0,1)$, the recognized language is
\[
L(P,\mu)=\{w\in\Sigma^\ast:f_P(w)>\mu\}.
\]
The end marker permits arbitrary stochastic postprocessing before the accepting set is read.

The quantum model used throughout is the one-way general quantum finite automaton (1gQFA) of Li et al. \cite{li2012}. Fix an $n$-dimensional Hilbert space $\mathcal H\cong\C^n$. A 1gQFA over $\Sigma$ is a tuple
\[
Q=(\mathcal H,\Sigma,\rho_0,\{\mathcal E_\sigma\}_{\sigma\in\Sigma},P_{\mathrm{acc}}),
\]
where $\rho_0$ is an initial density operator on $\mathcal H$, each $\mathcal E_\sigma$ is a completely positive trace-preserving map on $\mathcal L(\mathcal H)$, and $P_{\mathrm{acc}}$ is the projector onto the accepting subspace of $\mathcal H$. With $P_{\mathrm{rej}}=I-P_{\mathrm{acc}}$, the final measurement is the projective measurement $\{P_{\mathrm{acc}},P_{\mathrm{rej}}\}$. On input $w=\sigma_1\cdots\sigma_m$, the final state is
\[
\rho_w=\mathcal E_{\sigma_m}\circ\cdots\circ\mathcal E_{\sigma_1}(\rho_0),
\]
and the acceptance probability is
\[
f_Q(w)=\tr(P_{\mathrm{acc}}\rho_w).
\]
The number of states of $Q$ is $\dim\mathcal H=n$. Under a strict cutpoint $\lambda\in[0,1)$, the recognized language is
\[
L(Q,\lambda)=\{w\in\Sigma^\ast:f_Q(w)>\lambda\}.
\]

A standard mixed-state linearization represents the action of an $n$-state 1gQFA on $\operatorname{Herm}(\mathcal H)$ by real $n^2\times n^2$ matrices. Consequently, every $n$-state 1gQFA acceptance function has an $n^2$-dimensional real linear representation, equivalently an $n^2$-state GFA representation \cite{yakaryilmaz2011,li2012}.

\section{The Universal Upper Bound}\label{sec:upper}
The classical equivalence between GFAs and PFAs \cite{turakainen1969,paz1971,turakainen1975,hirvensaloyakaryilmaz2016,rote2025} provides a standard $(n^2+3)$-state simulation of an $n$-state 1gQFA. One state is saved by absorbing the cutpoint shift, while a second state is saved by reducing the stochastic overhead of the GFA-to-PFA conversion.

\subsection{A one-state zero-cutpoint GFA-to-PFA conversion}

\begin{proposition}\label{prop:zero-gfa-pfa}
Let $G$ be a $k$-state GFA over $\Sigma$. Then there is a one-way PFA $P$ over the same alphabet with at most $k+1$ states such that
\[
L(P,1/(k+1))=L(G,0).
\]
\end{proposition}

\begin{proof}
Write $G=(S,\Sigma,u,\{A_\sigma\}_{\sigma\in\Sigma},v)$. Without loss of generality, assume $v\ne0$. Set $m=k+1$, and let
\[
q=\frac1m(1,\ldots,1)\in\Delta^{m-1},
\qquad
W=\{z\in\R^{1\times m}:z\mathbf 1=0\}.
\]
The space $W$ has dimension $k$. The linear functionals $r\mapsto rv$ on $\R^{1\times k}$ and $z\mapsto ze_1$ on $W$ are both nonzero, so there is a linear isomorphism
\[
T:\R^{1\times k}\longrightarrow W
\]
such that
\begin{equation}\label{eq:T-pfa}
T(r)e_1=rv
\qquad(r\in\R^{1\times k}).
\end{equation}

For each $\sigma\in\Sigma$, define a linear operator $K_\sigma$ on $\R^{1\times m}=\operatorname{span}\{q\}\oplus W$ by
\[
qK_\sigma=0,
\qquad
T(r)K_\sigma=T(rA_\sigma).
\]
Its image lies in $W$, hence $K_\sigma\mathbf 1=0$. Let $Q$ be the $m\times m$ matrix whose every row is $q$. Then $qQ=q$ and $zQ=0$ for every $z\in W$. Since $Q$ is strictly positive and $\Sigma$ is finite, there is a common $\varepsilon>0$ small enough that
\[
S_\sigma=Q+\varepsilon K_\sigma
\]
is entrywise nonnegative for every $\sigma$. Because $K_\sigma\mathbf 1=0$, each $S_\sigma$ is row-stochastic.

Choose $\delta>0$ small enough that
\[
p_0=q+\delta T(u)\in\Delta^{m-1}.
\]
Let $P$ have initial distribution $p_0$, transition matrices $S_\sigma$, identity end marker, and accepting set $F=\{1\}$. If $A_w=A_{\sigma_1}\cdots A_{\sigma_\ell}$, induction on $|w|$ gives
\[
p_w=q+\delta\varepsilon^{|w|}T(uA_w),
\]
because $qS_\sigma=q$ and $T(r)S_\sigma=\varepsilon T(rA_\sigma)$. Therefore, by \cref{eq:T-pfa},
\[
\begin{aligned}
f_P(w)-\frac1m
&=\delta\varepsilon^{|w|}T(uA_w)e_1\\
&=\delta\varepsilon^{|w|}uA_wv\\
&=\delta\varepsilon^{|w|}f_G(w).
\end{aligned}
\]
The multiplicative factor is strictly positive, so the two sides have the same sign and $L(P,1/m)=L(G,0)$.
\end{proof}

The construction has a simple geometric interpretation. The $(k+1)$-state probability simplex provides exactly the $k$ independent directions needed to embed the GFA dynamics. The uniform distribution $q$ acts as an interior origin, and the scaling factor $\varepsilon$ shrinks the unconstrained GFA transitions into small perturbations around $q$ to keep them strictly inside the simplex. Although this scaling attenuates the final measurement by a strictly positive factor $\delta\varepsilon^{|w|}$, the sign of the zero-cutpoint computation is perfectly preserved.

\subsection{From 1gQFAs to PFAs}
For 1gQFAs, the cutpoint can be absorbed into the final functional while retaining the standard $n^2$-dimensional mixed-state representation.

\begin{theorem}\label{thm:upper}
Let $Q$ be an $n$-state 1gQFA and let $\lambda\in[0,1)$. Then there exist a one-way PFA $P$ over the same alphabet, with at most $n^2+1$ states, and a strict cutpoint $\mu\in[0,1)$ such that
\[
L(P,\mu)=L(Q,\lambda).
\]
\end{theorem}

\begin{proof}
By fixing a real basis $\{B_1, \ldots, B_{n^2}\}$ for $\operatorname{Herm}(\mathcal H)$, the standard mixed-state linearization \cite{yakaryilmaz2011,li2012} represents an $n$-state 1gQFA $Q$ by an $n^2$-state GFA
\[
G_Q=(S,\Sigma,u,\{A_\sigma\}_{\sigma\in\Sigma},v),
\]
where the standard choice $v_i=\tr(P_{\mathrm{acc}}B_i)$ gives $f_{G_Q}(w)=f_Q(w)$. Since every reachable density operator satisfies $\tr(\rho_w)=1$, we instead choose
\[
v_i=\tr\bigl((P_{\mathrm{acc}}-\lambda I)B_i\bigr),
\]
which gives $f_{G_Q}(w)=f_Q(w)-\lambda$. Hence
\[
L(G_Q,0)=L(Q,\lambda),
\]
with the same $n^2$ states and the same alphabet.

Applying \Cref{prop:zero-gfa-pfa} with $k=n^2$ gives a PFA with at most $n^2+1$ states recognizing the same language at cutpoint $1/(n^2+1)$.
\end{proof}

The trace-one identity is what makes the improvement possible: it moves the quantum cutpoint to zero inside the existing mixed-state representation, after which the one-state stochasticization applies directly. Hence
\begin{equation}\label{eq:C-upper}
C_n\le n^2+1.
\end{equation}
The fixed-alphabet construction in \Cref{sec:lower} attains this bound.

\section{Dynamic Shattering and Lower Bound Construction}\label{sec:lower}
We first introduce two notions of shattering for deriving state lower bounds. Shattering uses suffixes to realize every possible acceptance/rejection pattern on a selected set of prefixes, translating language distinctions into constraints on any automaton recognizing the language. Dynamic shattering strengthens this requirement by imposing persistent changes across the cutpoint, thereby improving the corresponding PFA lower bound by one state. Using these tools, we then construct a carefully designed fixed-alphabet language recognized by an $n$-state 1gQFA and prove the exact lower bound $n^2+1$.

\begin{definition}\label{def:shattering}
Let $L\subseteq\Sigma^\ast$. The language $L$ \emph{shatters} prefixes $x_1,\ldots,x_N$ if, for every label vector $\ell=(\ell_1,\ldots,\ell_N)\in\{\pm1\}^N$, there is a suffix $y_\ell$ such that
\[
x_i y_\ell\in L\iff \ell_i=+1
\qquad(1\le i\le N).
\]
For a letter $\mathtt{a}\in\Sigma$, the language $L$ \emph{dynamically shatters} these prefixes with respect to $\mathtt{a}$ if the suffixes $y_\ell$ can be chosen so that, in addition, for every $\ell$ and every integer $d\ge1$, the sequence
\[
\mathtt{a}^{qd}y_\ell,\qquad q=0,1,2,\ldots,
\]
contains infinitely many words in $L$ and infinitely many words outside $L$.
\end{definition}

\begin{lemma}\label{lem:shattering}
Let $L\subseteq\Sigma^\ast$ and let $x_1,\ldots,x_N$ be prefixes.
\begin{enumerate}
\item If $L$ shatters $x_1,\ldots,x_N$, then every PFA recognizing $L$ under a strict cutpoint has at least $N$ states.
\item If $L$ dynamically shatters $x_1,\ldots,x_N$ with respect to some letter $\mathtt{a}$, then every PFA recognizing $L$ under a strict cutpoint has at least $N+1$ states.
\end{enumerate}
\end{lemma}

\begin{proof}
Let an $m$-state PFA recognize $L$ at cutpoint $\mu$. For each label vector $\ell$, choose a shattering suffix $y_\ell$. Let $\delta_i$ be the row distribution reached after $x_i$, and absorb $y_\ell$ and the right end marker into the column
\[
b_\ell=P_{y_\ell}P_{\#}\mathbf1_F\in[0,1]^m.
\]
Set $c_\ell=b_\ell-\mu\mathbf1$. The shattering condition gives
\[
\ell_i=+1\Rightarrow \delta_i c_\ell>0,
\qquad
\ell_i=-1\Rightarrow \delta_i c_\ell\le0.
\]
Suppose that $\sum_i\alpha_i\delta_i=0$ is a nontrivial linear relation. Reversing all coefficients if necessary, assume that some $\alpha_i>0$, and choose $\ell_i=+1$ exactly when $\alpha_i>0$. Then $\alpha_i(\delta_i c_\ell)\ge0$ for every $i$, with strict inequality whenever $\alpha_i>0$. Right-multiplying the relation by $c_\ell$ therefore gives
\[
0=\sum_i\alpha_i(\delta_i c_\ell)>0,
\]
a contradiction. Hence $\delta_1,\ldots,\delta_N$ are linearly independent in $\mathbb R^m$, proving $m\ge N$.

Now further assume dynamic shattering with respect to $\mathtt a$. Let $p_0$ be the initial row distribution and let $S$ be the transition matrix for $\mathtt a$. The finite Markov-chain decomposition gives an integer $d\ge1$ for which $S^{qd}$ converges as $q\to\infty$: take $d$ to be a common multiple of the periods of the recurrent classes; the transient part then decays, and each recurrent class converges along the $d$-step subsequence. Define
\[
\omega=\lim_{q\to\infty}p_0S^{qd}.
\]
For each $\ell$,
\[
p_0S^{qd}b_\ell=f_P(\mathtt{a}^{qd}y_\ell)\longrightarrow\omega b_\ell.
\]
Dynamic shattering makes this convergent sequence exceed $\mu$ infinitely often and stay at most $\mu$ infinitely often, so its limit is $\mu$. Thus $\omega c_\ell=0$ for every $\ell$.

Set $z_i=\delta_i-\omega$. Then
\[
z_i\mathbf1=0,
\qquad
z_i c_\ell=\delta_i c_\ell.
\]
Hence the $z_i$ satisfy the same sign-separation relations as the $\delta_i$, so the preceding sign argument shows that $z_1,\ldots,z_N$ are linearly independent. They all lie in
\[
W=\{z\in\mathbb R^{1\times m}:z\mathbf1=0\},
\qquad \dim W=m-1,
\]
so $N\le m-1$ and therefore $m\ge N+1$.
\end{proof}

Dynamic shattering is naturally realizable by quantum automata because unitary evolution can sustain non-decaying oscillations, whereas finite probabilistic dynamics becomes convergent after passing to a suitable arithmetic subsequence. This contrast is what yields the additional one-state improvement. The quadratic scale of the construction comes primarily from the real dimension $n^2-1$ of the traceless Hermitian operator space, while dynamic shattering forces the final additional state.

Fix $n\ge2$ and set
\[
D=n^2-1,
\qquad
q=\lfloor D/2\rfloor,
\qquad
\Sigma_4=\{\mathtt{a},\lp,\lt,\mathtt{h}\}.
\]
Let $e_1,\ldots,e_D$ be the standard basis of $\mathbb R^D$, and let $J$ be the nilpotent shift $Je_i=e_{i+1}$ for $i<D$ and $Je_D=0$. Define $R\in O(D)$ by rotating $\operatorname{span}\{e_{2j-1},e_{2j}\}$ through angle $\pi/4^j$ for $1\le j\le q$ and, when $D$ is odd, acting as $-1$ on $e_D$. Let $u\in\mathbb R^D$ have unit projection on the $j$-th rotation plane at angle $\pi/(2\cdot4^j)$,
\[
u_{2j-1}=\cos\frac{\pi}{2\cdot4^j},
\qquad
u_{2j}=\sin\frac{\pi}{2\cdot4^j},
\]
and put $u_D=1$ when $D$ is odd.

\begin{lemma}\label{lem:orbit}
For every $s\in\{\pm1\}^D$, there is an integer $N_s\ge0$ such that
\[
\sgn\bigl((R^{N_s}u)_i\bigr)=s_i
\qquad(1\le i\le D).
\]
\end{lemma}

\begin{proof}
Number the open quadrants counterclockwise from the first by $0,1,2,3$. In the $j$-th plane, the angle of $R^Nu$ is
\[
\frac{\pi}{4^j}\left(N+\frac12\right)
=\frac{\pi}{2}\frac{2N+1}{4^j}.
\]
Because $2N+1$ is odd, the point never lies on a coordinate axis, and its quadrant is the $j$-th base-$4$ digit of $2N+1$. Choose digits $d_j\in\{0,1,2,3\}$ whose quadrants realize the prescribed sign pairs. If $D$ is odd, choose the units digit $d_0=1$ for $s_D=+1$ and $d_0=3$ for $s_D=-1$; if $D$ is even, take $d_0=1$. With
\[
K=\sum_{j=0}^{q}d_j4^j,
\qquad
N_s=\frac{K-1}{2},
\]
all planar signs are the prescribed ones, and in the odd-dimensional case $(-1)^{N_s}=s_D$.
\end{proof}

\begin{theorem}\label{thm:lower}
For every $n\ge2$, there exist an $n$-state 1gQFA $Q_n$ over the fixed alphabet $\Sigma_4$ and a strict cutpoint $\lambda_n$ such that $L(Q_n,\lambda_n)$ shatters $n^2$ prefixes. The same prefixes are dynamically shattered with respect to $\mathtt{a}$. Consequently, every PFA recognizing $L(Q_n,\lambda_n)$ under a strict cutpoint has at least $n^2+1$ states.
\end{theorem}

\begin{proof}
Let $\mathcal H=\mathbb C^n$, $\tau=I/n$, and
\[
V=\{H\in\operatorname{Herm}(\mathcal H):\tr H=0\},
\qquad \dim V=D.
\]
On $\operatorname{span}\{\ket1,\ket2\}$, set
\[
X=\ket1\bra2+\ket2\bra1,
\qquad
Y=-i\ket1\bra2+i\ket2\bra1,
\]
and define
\[
\ket\psi=\frac{\ket1+e^{i\pi/8}\ket2}{\sqrt2},
\qquad
P_{\mathrm{acc}}=\ket\psi\bra\psi,
\qquad
\lambda_0=\tr(P_{\mathrm{acc}}\tau)=\frac1n.
\]
The functional $\zeta(H)=\tr(P_{\mathrm{acc}}H)$ satisfies
\[
\zeta(X)=\cos\frac\pi8=u_1,
\qquad
\zeta(Y)=\sin\frac\pi8=u_2.
\]
Because $\zeta(X)=u^\top e_1$ and $\zeta(Y)=u^\top e_2$, the assignment $e_1\mapsto X$ and $e_2\mapsto Y$ equates the nontrivial functionals $x\mapsto u^\top x$ and $H\mapsto\zeta(H)$ on $\operatorname{span}\{e_1,e_2\}$. Because both functionals have $(D-1)$-dimensional kernels, this assignment extends to a full real-linear isomorphism $T:\mathbb R^D\to V$ that intertwines them everywhere, yielding
\begin{equation}\label{eq:T-fixed}
T(e_1)=X,
\qquad
T(e_2)=Y,
\qquad
\zeta(Tx)=u^\top x.
\end{equation}

Define
\[
L_{\lp}=TJT^{-1},
\qquad
L_{\lt}=TR^\top T^{-1}
\]
on $V$, and let $\widehat L_{\lp},\widehat L_{\lt}$ be their complex-linear extensions to the traceless operators. For $\sigma\in\{\lp,\lt\}$, set
\[
\mathcal E_\sigma(A)
=\tr(A)\tau
+\varepsilon\widehat L_\sigma\bigl(A-\tr(A)\tau\bigr).
\]
Since $L_\sigma$ maps Hermitian operators to Hermitian operators, its complex-linear extension $\widehat L_\sigma$ is Hermiticity-preserving. At $\varepsilon=0$, $\mathcal E_\sigma$ is the replacement channel $A\mapsto\tr(A)\tau,$
whose Choi matrix $I\otimes\tau$ is positive definite. The Choi matrix of $\mathcal E_\sigma$ is therefore Hermitian and depends continuously on $\varepsilon$, so it remains positive definite for a common sufficiently small $0<\varepsilon<1$. Hence $\mathcal E_\sigma$ is completely positive~\cite{choi1975}; trace preservation follows because $\widehat L_\sigma$ maps traceless operators to traceless operators. For this choice,
\begin{equation}\label{eq:pt-fixed}
\mathcal E_{\lp}(\tau+T(x))=\tau+\varepsilon T(Jx),
\qquad
\mathcal E_{\lt}(\tau+T(x))=\tau+\varepsilon T(R^\top x).
\end{equation}
Choose $0<\eta<1/n$ and take $\rho_0=\tau+\eta X$, which is positive definite.

Let $\theta/(2\pi)$ be irrational and define $\mathcal E_{\mathtt a}(\rho)=U\rho U^\dagger$, where
\[
U=\operatorname{diag}(e^{-i\theta/2},e^{i\theta/2},1,\ldots,1).
\]
Then $UXU^\dagger=\cos\theta\,X+\sin\theta\,Y$, so
\begin{equation}\label{eq:a-orbit}
\rho_{\mathtt{a}^k}
=\tau+\eta T\bigl(\cos(k\theta)e_1+\sin(k\theta)e_2\bigr).
\end{equation}

For $1\le i\le D$, let $x_i=\lp^{i-1}$ and let $x_0=\lp^D$. Since
$\rho_0=\tau+\eta T(e_1)$, repeated application of \cref{eq:pt-fixed} gives,
for every $s\in\{\pm1\}^D$,
\[
\mathcal E_{\lt}^{N_s}\!\left(
\mathcal E_{\lp}^{\,i-1}(\rho_0)\right)
=
\tau+\eta\varepsilon^{i-1+N_s}
T\!\left((R^\top)^{N_s}e_i\right)
\qquad(1\le i\le D),
\]
whereas
\[
\mathcal E_{\lt}^{N_s}\!\left(
\mathcal E_{\lp}^{\,D}(\rho_0)\right)=\tau
\]
because $J^De_1=0$. Taking the accepting measurement and using
\cref{eq:T-fixed} therefore gives
\[
\tr\!\left(
P_{\mathrm{acc}}\,
\mathcal E_{\lt}^{N_s}
\bigl(\mathcal E_{\lp}^{\,i-1}(\rho_0)\bigr)
\right)-\lambda_0
=
\eta\varepsilon^{i-1+N_s}(R^{N_s}u)_i.
\]
Likewise, \cref{eq:T-fixed,eq:pt-fixed,eq:a-orbit} gives
\[
\tr\!\left(
P_{\mathrm{acc}}\,
\mathcal E_{\lt}^{N_s}
\bigl(\mathcal E_{\mathtt a}^{\,k}(\rho_0)\bigr)
\right)-\lambda_0
=
\eta\varepsilon^{N_s}
\left[
(R^{N_s}u)_1\cos(k\theta)
+
(R^{N_s}u)_2\sin(k\theta)
\right].
\]
The sinusoid on the right has amplitude $\eta\varepsilon^{N_s}$, since
the projection of $R^{N_s}u$ onto the first rotation plane has norm $1$.

By \Cref{lem:orbit}, every $(R^{N_s}u)_i$ is nonzero. Since only finitely
many $s$ occur, choose $0<\delta<1-\lambda_0$ so small that
\[
\delta<
\eta\varepsilon^{i-1+N_s}|(R^{N_s}u)_i|
\quad
\text{for all $s\in\{\pm1\}^D$ and $1\le i\le D$},
\]
and
\[
\delta<\eta\varepsilon^{N_s}
\quad
\text{for all $s\in\{\pm1\}^D$}.
\]
Set $\lambda_n=\lambda_0+\delta$.

By construction, the acceptance probability of the anchor prefix $x_0$ is strictly below the cutpoint $\lambda_n$ for every suffix $\lt^{N_s}$. Thus, we introduce a fourth letter $\mathtt{h}$ that inverts the acceptance probability across $\lambda_n$. Since $0<\lambda_n<1$, choose
$\gamma>0$ sufficiently small that
\[
H=\lambda_n(1+\gamma)I-\gamma P_{\mathrm{acc}}
\]
satisfies $0<H<I$. Let
\[
\ket{\psi_\perp}
=
\frac{\ket1-e^{i\pi/8}\ket2}{\sqrt2}
\]
and define
\[
\mathcal E_{\mathtt h}(\rho)
=
\tr(H\rho)\ket\psi\bra\psi
+
\tr((I-H)\rho)\ket{\psi_\perp}\bra{\psi_\perp}.
\]
This measure-and-prepare map is a quantum channel. Hence
\[
Q_n=
(\mathcal H,\Sigma_4,\rho_0,
\{\mathcal E_{\mathtt a},\mathcal E_{\lp},
\mathcal E_{\lt},\mathcal E_{\mathtt h}\},
P_{\mathrm{acc}})
\]
is a well-defined $n$-state 1gQFA. Since $\ket\psi$ is accepted with
probability $1$ and $\ket{\psi_\perp}$ with probability $0$, every word
$w$ satisfies
\begin{equation}\label{eq:exact-flip}
f_{Q_n}(w\mathtt h)-\lambda_n
=
-\gamma\bigl(f_{Q_n}(w)-\lambda_n\bigr).
\end{equation}

By our choice of $\delta$, the signs of the static deviations yield
\begin{equation}\label{eq:anchored-halfcube}
x_i\lt^{N_s}\in L(Q_n,\lambda_n)
\iff s_i=+1
\quad(1\le i\le D),
\qquad
x_0\lt^{N_s}\notin L(Q_n,\lambda_n),
\end{equation}
and
\begin{equation}\label{eq:dynamic-fixed}
f_{Q_n}(\mathtt a^k\lt^{N_s})-\lambda_n
=
\eta\varepsilon^{N_s}
\left[
(R^{N_s}u)_1\cos(k\theta)
+
(R^{N_s}u)_2\sin(k\theta)
\right]
-\delta.
\end{equation}
The sinusoidal amplitude in \cref{eq:dynamic-fixed} exceeds $\delta$.
For every $d\ge1$, $d\theta/(2\pi)$ is irrational, so the phases
$qd\theta$ are dense modulo $2\pi$. Hence
\cref{eq:dynamic-fixed} is positive infinitely often and negative
infinitely often along $k=qd$.

For
$\ell=(\ell_0,\ell_1,\ldots,\ell_D)\in\{\pm1\}^{D+1}$, define
\[
y_\ell=
\begin{cases}
\lt^{N_{(\ell_1,\ldots,\ell_D)}},
&\ell_0=-1,\\
\lt^{N_{(-\ell_1,\ldots,-\ell_D)}}\mathtt h,
&\ell_0=+1.
\end{cases}
\]
Then \cref{eq:anchored-halfcube,eq:exact-flip} give
\[
x_i y_\ell\in L(Q_n,\lambda_n)
\iff \ell_i=+1
\qquad(0\le i\le D).
\]
Thus the $D+1=n^2$ prefixes are shattered in the sense of
\Cref{def:shattering}. By the first part of \Cref{lem:shattering}, this static part of the
construction alone already forces at least $n^2$ states in any PFA
recognizing $L(Q_n,\lambda_n)$ under a strict cutpoint.

\Cref{eq:dynamic-fixed,eq:exact-flip} also show that every $y_\ell$
has the required oscillation along every arithmetic progression.
Hence the same $n^2$ prefixes are dynamically shattered with respect to
$\mathtt a$. By the second part of \Cref{lem:shattering}, the lower bound strengthens
from $n^2$ to $n^2+1$.
\end{proof}

\begin{corollary}\label{cor:main}
$C_1=1$, and for every $n\ge2$,
\[
C_n=n^2+1.
\]
For $n\ge2$, the lower bound is attained over the same four-letter alphabet $\Sigma_4$.
\end{corollary}

\section{Conclusion}\label{sec:conclusion}

This paper determines the exact worst-case classical simulation cost of one-way general quantum finite automata under strict cutpoints. The improved zero-cutpoint conversion gives the upper bound $C_n\le n^2+1$, while the shattering construction gives the matching lower bound, resolving the open problem with $C_n=n^2+1$. The exact lower bound is realized over a fixed four-letter alphabet. It remains open whether $n^2+1$ can be attained over a smaller alphabet. Moreover, the cutpoints $\lambda_n$ in \Cref{thm:lower} tend to zero, and it is unknown whether the same exact lower bound can be obtained with a fixed cutpoint independent of $n$. These questions concern refinements of the extremal construction rather than the exact state complexity established here.

A natural next direction is to determine analogous exact simulation costs for other one-way QFA models. Hybrid quantum--classical models require the classical and quantum state resources to be treated jointly. For more restricted models such as measure-once QFAs, the relevant state-vector geometry has only linear dimension, so it is unclear whether the quadratic simulation cost persists. Resolving such cases may require substantially different techniques.

Finally, shattering and dynamic shattering are language-theoretic notions that may be useful beyond the present 1gQFA setting. Their prefix--suffix viewpoint is closely related to classical rank-based methods for automata and to sign-rank techniques in unbounded-error communication complexity \cite{forster2002}. Whether these viewpoints admit a deeper and more useful connection remains open.

\printbibliography

@article{rabin1963,
  author  = {Michael O. Rabin},
  title   = {Probabilistic Automata},
  journal = {Information and Control},
  volume  = {6},
  number  = {3},
  pages   = {230--245},
  year    = {1963},
  doi     = {10.1016/S0019-9958(63)90290-0}
}

@book{paz1971,
  author    = {Azaria Paz},
  title     = {Introduction to Probabilistic Automata},
  publisher = {Academic Press},
  year      = {1971}
}

@article{turakainen1969,
  author  = {Paavo Turakainen},
  title   = {Generalized Automata and Stochastic Languages},
  journal = {Proceedings of the American Mathematical Society},
  volume  = {21},
  number  = {2},
  pages   = {303--309},
  year    = {1969},
  doi     = {10.1090/S0002-9939-1969-0242596-1}
}

@inproceedings{kondacs1997,
  author    = {Attila Kondacs and John Watrous},
  title     = {On the Power of Quantum Finite State Automata},
  booktitle = {Proceedings of the 38th Annual Symposium on Foundations of Computer Science},
  pages     = {66--75},
  publisher = {IEEE Computer Society},
  year      = {1997},
  doi       = {10.1109/SFCS.1997.646094}
}

@article{li2012,
  author  = {Lvzhou Li and Daowen Qiu and Xiangfu Zou and Lvjun Li and Lihua Wu and Paulo Mateus},
  title   = {Characterizations of One-Way General Quantum Finite Automata},
  journal = {Theoretical Computer Science},
  volume  = {419},
  pages   = {73--91},
  year    = {2012},
  doi     = {10.1016/j.tcs.2011.10.021}
}

@inproceedings{yakaryilmaz2010,
  author    = {Abuzer Yakary{\i}lmaz and A. C. Cem Say},
  title     = {Languages Recognized with Unbounded Error by Quantum Finite Automata},
  booktitle = {Computer Science -- Theory and Applications},
  series    = {Lecture Notes in Computer Science},
  volume    = {5675},
  pages     = {356--367},
  year      = {2009},
  publisher = {Springer},
  doi       = {10.1007/978-3-642-03351-3_33}
}

@article{yakaryilmaz2011,
  author  = {Abuzer Yakary{\i}lmaz and A. C. Cem Say},
  title   = {Unbounded-Error Quantum Computation with Small Space Bounds},
  journal = {Information and Computation},
  volume  = {209},
  number  = {6},
  pages   = {873--892},
  year    = {2011},
  doi     = {10.1016/j.ic.2011.01.008}
}

@article{moore2000,
  author  = {Cristopher Moore and James P. Crutchfield},
  title   = {Quantum Automata and Quantum Grammars},
  journal = {Theoretical Computer Science},
  volume  = {237},
  number  = {1--2},
  pages   = {275--306},
  year    = {2000},
  doi     = {10.1016/S0304-3975(98)00191-1}
}

@inproceedings{ambainis1998,
  author    = {Andris Ambainis and Rūsiņš Freivalds},
  title     = {1-Way Quantum Finite Automata: Strengths, Weaknesses and Generalizations},
  booktitle = {Proceedings of the 39th Annual Symposium on Foundations of Computer Science},
  pages     = {332--341},
  publisher = {IEEE Computer Society},
  year      = {1998},
  doi       = {10.1109/SFCS.1998.743469}
}

@article{brodsky2002,
  author  = {Alex Brodsky and Nicholas Pippenger},
  title   = {Characterizations of 1-Way Quantum Finite Automata},
  journal = {SIAM Journal on Computing},
  volume  = {31},
  number  = {5},
  pages   = {1456--1478},
  year    = {2002},
  doi     = {10.1137/S0097539799353443}
}

@article{turakainen1975,
  author  = {Paavo Turakainen},
  title   = {Word-Functions of Stochastic and Pseudo Stochastic Automata},
  journal = {Annales Academiae Scientiarum Fennicae. Series A I. Mathematica},
  volume  = {1},
  number  = {1},
  pages   = {27--37},
  year    = {1975},
  doi     = {10.5186/aasfm.1975.0126}
}

@article{hirvensaloyakaryilmaz2016,
  author  = {Mika Hirvensalo and Abuzer Yakary{\i}lmaz},
  title   = {Decision Problems on Unary Probabilistic and Quantum Automata},
  journal = {Baltic Journal of Modern Computing},
  volume  = {4},
  number  = {4},
  pages   = {965--976},
  year    = {2016},
  doi     = {10.22364/bjmc.2016.4.4.22}
}

@incollection{sayyakaryilmaz2014,
  author    = {A. C. Cem Say and Abuzer Yakary{\i}lmaz},
  title     = {Quantum Finite Automata: A Modern Introduction},
  booktitle = {Computing with New Resources: Essays Dedicated to Jozef Gruska on the Occasion of His 80th Birthday},
  series    = {Lecture Notes in Computer Science},
  volume    = {8808},
  pages     = {208--222},
  publisher = {Springer},
  year      = {2014},
  doi       = {10.1007/978-3-319-13350-8_16}
}

@incollection{ambainisyakaryilmaz2021,
  author    = {Andris Ambainis and Abuzer Yakary{\i}lmaz},
  title     = {Automata and Quantum Computing},
  booktitle = {Handbook of Automata Theory},
  volume    = {2},
  pages     = {1457--1493},
  publisher = {EMS Press},
  year      = {2021},
  doi       = {10.4171/AUTOMATA-2/17}
}

@article{zheng2014,
  author  = {Shenggen Zheng and Jozef Gruska and Daowen Qiu},
  title   = {On the State Complexity of Semi-Quantum Finite Automata},
  journal = {RAIRO -- Theoretical Informatics and Applications},
  volume  = {48},
  number  = {2},
  pages   = {187--207},
  year    = {2014},
  doi     = {10.1051/ita/2014003}
}

@article{qiu2015,
  author  = {Daowen Qiu and Lvzhou Li and Paulo Mateus and Amilcar Sernadas},
  title   = {Exponentially More Concise Quantum Recognition of Non-{RMM} Regular Languages},
  journal = {Journal of Computer and System Sciences},
  volume  = {81},
  number  = {2},
  pages   = {359--375},
  year    = {2015},
  doi     = {10.1016/j.jcss.2014.06.008}
}

@article{mereghetti2020,
  author  = {Carlo Mereghetti and Beatrice Palano and Simone Cialdi and Valeria Vento and Matteo G. A. Paris and Stefano Olivares},
  title   = {Photonic Realization of a Quantum Finite Automaton},
  journal = {Physical Review Research},
  volume  = {2},
  number  = {1},
  eid     = {013089},
  year    = {2020},
  doi     = {10.1103/PhysRevResearch.2.013089}
}

@article{mishra2020,
  author  = {Nimish Mishra and {Rayala Sarath Chandra} and Bikash K. Behera and Prasanta K. Panigrahi},
  title   = {Automation of Quantum Braitenberg Vehicles Using Finite Automata: Moore Machines},
  journal = {Quantum Information Processing},
  volume  = {19},
  number  = {1},
  eid     = {17},
  year    = {2020},
  doi     = {10.1007/s11128-019-2512-2}
}

@article{xiaoqiu2025,
  author  = {Ligang Xiao and Daowen Qiu},
  title   = {State Complexity of One-Way Quantum Finite Automata Together with Classical States},
  journal = {Journal of Computer and System Sciences},
  volume  = {154},
  eid     = {103659},
  year    = {2025},
  doi     = {10.1016/j.jcss.2025.103659}
}

@book{watrous2018,
  author    = {John Watrous},
  title     = {The Theory of Quantum Information},
  publisher = {Cambridge University Press},
  year      = {2018},
  doi       = {10.1017/9781316848142}
}

@article{choi1975,
  author  = {Man-Duen Choi},
  title   = {Completely Positive Linear Maps on Complex Matrices},
  journal = {Linear Algebra and its Applications},
  volume  = {10},
  number  = {3},
  pages   = {285--290},
  year    = {1975},
  doi     = {10.1016/0024-3795(75)90075-0}
}

@inproceedings{rote2025,
  author    = {G{\"u}nter Rote},
  title     = {Probabilistic Finite Automaton Emptiness Is Undecidable for a Fixed Automaton},
  booktitle = {50th International Symposium on Mathematical Foundations of Computer Science (MFCS 2025)},
  series    = {Leibniz International Proceedings in Informatics},
  volume    = {345},
  pages     = {86:1--86:18},
  publisher = {Schloss Dagstuhl -- Leibniz-Zentrum f{\"u}r Informatik},
  year      = {2025},
  doi       = {10.4230/LIPIcs.MFCS.2025.86}
}

@article{forster2002,
  author  = {J{\"u}rgen Forster},
  title   = {A Linear Lower Bound on the Unbounded Error Probabilistic Communication Complexity},
  journal = {Journal of Computer and System Sciences},
  volume  = {65},
  number  = {4},
  pages   = {612--625},
  year    = {2002},
  doi     = {10.1016/S0022-0000(02)00019-3}
}

@inproceedings{diazcaroyakaryilmaz2016,
  author    = {Alejandro D{\'i}az-Caro and Abuzer Yakary{\i}lmaz},
  title     = {Affine Computation and Affine Automaton},
  booktitle = {Computer Science -- Theory and Applications},
  series    = {Lecture Notes in Computer Science},
  volume    = {9691},
  pages     = {146--160},
  publisher = {Springer},
  year      = {2016},
  doi       = {10.1007/978-3-319-34171-2_11}
}

\end{document}